\documentclass[11pt]{article}

\usepackage{jeffstyle}
\usepackage{cite}
\usepackage{pbox}

\newcommand{\op}{\ensuremath{{\perp}}}

\title{Relative Errors for Deterministic Low-Rank Matrix Approximations}
\author{
Mina Ghashami\\ University of Utah \\ \texttt{ghashami@cs.utah.edu}
\and
Jeff M. Phillips\\ University of Utah \\ \texttt{jeffp@cs.utah.edu}
}

\begin{document}
\begin{titlepage}
\maketitle

\begin{abstract}
We consider processing an $n \times d$ matrix $A$ in a stream with row-wise updates according to a recent algorithm called Frequent Directions (Liberty, KDD 2013).  This algorithm maintains an $\ell \times d$ matrix $Q$ deterministically, processing each row in $O(d \ell^2)$ time; the processing time can be decreased to $O(d\ell)$ with a slight modification in the algorithm and a constant increase in space.  Then for any unit vector $x$, the matrix $Q$ satisfies
\[
0 \leq \|A x\|^2 - \|Qx\|^2 \leq \|A\|_F^2/\ell.  
\]
We show that if one sets $\ell = \lceil k+ k/\eps \rceil$ and returns $Q_k$, a $k \times d$ matrix that is simply the top $k$ rows of $Q$, then we achieve the following properties:
\[
\|A - A_k\|_F^2 \leq \|A\|_F^2 - \|Q_k\|_F^2 \leq (1+\eps) \|A - A_k\|_F^2
\]
and where $\pi_{Q_k}(A)$ is the projection of $A$ onto the rowspace of $Q_k$ then
\[
\|A - \pi_{Q_k}(A)\|_F^2 \leq (1+\eps) \|A - A_k\|_F^2.
\]
  
We also show that Frequent Directions cannot be adapted to a sparse version in an obvious way that retains $\ell$ original rows of the matrix, as opposed to a linear combination or sketch of the rows.  
\end{abstract}
\end{titlepage}

\section{Introduction}

The data streaming paradigm~\cite{muthukrishnan05:_book} considers computation on a large data set $A$ where one data item arrives at a time, is processed, and then is not read again.  It enforces that only a small amount of memory is available at any given time.  This small space constraint is critical when the full data set cannot fit in memory or disk.  Typically, the amount of space required is traded off with the accuracy of the computation on $A$.  Usually the computation results in some summary $S(A)$ of $A$, and this trade-off determines how accurate one can be with the available space resources.  Although computational runtime is important, in this paper we mainly focus on space constraints and the types of approximation guarantees that can be made.  

In truly large datasets, one processor (and memory) is often incapable of handling all of the dataset $A$ in a feasible amount of time.  Even reading a terabyte of data on a single processor can take many hours.  Thus this computation is often spread among some set of processors, and then the summary of $A$ is combined after (or sometimes during~\cite{Cor13}) its processing on each processor.  Again often each item is read once, whether it comes from a single large source or is being generated on the fly.  The key computational problem shifts from updating a summary $S(A)$ when witnessing a single new data item (the streaming model), to taking two summaries $S(A_1)$ and $S(A_2)$ and constructing a single new summary $S(A_1 \cup A_2)$.  In this new paradigm the goal is to have the same space-approximation trade-offs in $S(A_1 \cup A_2)$ as possible for a streaming algorithm.  When such a process is possible, the summary is known as a \emph{mergeable summary}~\cite{ACHPWY12}.  Linear sketches are trivially mergeable, so this allows many streaming algorithms to directly translate to this newer paradigm.  
Again, space is a critical resource since it directly corresponds with the amount of data needed to transmit across the network, and emerging cost bottleneck in big data systems.  

In this paper we focus on \emph{deterministic} mergeable summaries for low-rank matrix approximation, based on recent work by Liberty~\cite{Lib12}, that is already known to be mergeable~\cite{Lib12}.  Thus our focus is a more careful analysis of the space-error trade off for the algorithm, and we describe them under the streaming setting for simplicity; all bounds directly carry over into mergeable summary results.  

In particular we re-analyze the Frequent Directions algorithm of Liberty to show it provides relative error bounds for matrix sketching, and conjecture it achieves the optimal space, up to $\log$ factors, for any row-update based summary.  
This supports the strong empirical results of Liberty~\cite{Lib12}.  
His analysis only provided additive error bounds which are hard to compare to more conventional ways of measuring accuracy of matrix approximation algorithms.

\subsection{Problem Statement and Related Work}

In this problem $A$ is an $n \times d$ matrix and the stream processes each row $a_i$ (of length $d$) at a time.  Typically the matrix is assumed to be \emph{tall} so $n \gg d$, and sometimes the matrix will be assumed to be \emph{sparse} so the number of non-zero entries $\nnz(A)$ of $A$ will be small, $\nnz(A) \ll nd$ (e.g. $\nnz(A) = O((n+d)\log(nd)))$.  

The best rank-$k$ approximation to $A$ (under Frobenius or 2 norm) is denoted as $A_k$ and can be computed in $O(nd^2)$ time on a tall matrix using the singular value decomposition.  The $\svd(A)$ produces three matrices $U$, $S$, and $V$ where $U$ and $V$ are orthonormal, of size $n \times n$ and $d \times d$, respectively, and $S$ is $n \times d$ but only has non-zero elements on its diagonal $\{\sigma_1, \ldots, \sigma_d\}$.  
Let $U_k$, $S_k$, and $V_k$ be the first $k$ columns of each matrix, then $A = U S V^T$ and $A_k = U_k S_k V_k^T$.  
Note that although $A_k$ requires $O(nd)$ space, the set of matrices $\{U_k, S_k,V_k\}$ require only a total of $O((n+d)k)$ space (or $O(nk)$ if the matrix is tall).  Moreover, even the set $\{U, S, V\}$ really only takes $O(nd+d^2)$ space since we can drop the last $n-d$ columns of $U$, and the last $n-d$ rows of $S$ without changing the result.  
In the streaming version, the goal is to compute something that replicates the effect of $A_k$ using less space and only seeing each row once.  

\paragraph{Construction bounds.}
The strongest version, (providing \emph{construction} bounds) for some parameter $\eps\in(0,1)$, is some representation of a rank $k$ matrix $\hat A$ such that $\|A - \hat A\|_\xi \leq (1+\eps) \|A - A_k\|_\xi$ for $\xi = \{2,F\}$.  Unless $A$ is sparse, then storing $\hat A$ explicitly may require $\Omega(nd)$ space, so that is why various representations of $\hat A$ are used in its place.  This can include decompositions similar to the SVD, e.g.  a CUR decomposition~\cite{drineas2003pass,mahoney2009cur,drineas2008relative} where $\hat A = C U R$ and where $U$ is small and dense, and $C$ and $R$ are sparse and skinny, or others~\cite{clarkson2013low} where the middle matrix is still diagonal.  The sparsity is often preserved by constructing the wrapper matrices (e.g. $C$ and $R$) from the original columns or rows of $A$.  
There is an obvious $\Omega(n + d)$ space bound for any construction result in order to preserve the column and the row space.  

\paragraph{Projection bounds.}
Alternatively, a weaker version (providing \emph{projection} bounds) just finds a rank $k$ subspace $B_k$ where the projection of $A$ onto this subspace $\pi_{B_k}(A)$ represents $\hat A$.  This bound is weaker since this cannot actually represent $\hat A$ without making another pass over $A$ to do the projection.  An even weaker version finds a rank $r > k$ subspace $B$, where $\hat A$ is represented by the best rank $k$ approximation of $\pi_B(A)$; note that $\pi_B(A)$ is then also rank $r$, not $k$.  
However, when $B$ or $B_k$ is composed of a set of $\ell$ rows (and perhaps $B_k$ is only $k$ rows) then the total size is only $O(d \ell)$ (allotting constant space for each entry); so it does not depend on $n$.  This is a significant advantage in tall matrices where $n \gg d$.  
Sometimes this subspace approximation is sufficient for downstream analysis, since the rowspace is still (approximately) preserved.  For instance, in PCA the goal is to compute the most important directions in the row space.  

\paragraph{Streaming algorithms.}
Many of these algorithms are \emph{not} streaming algorithms.  To the best of our understanding, the best streaming algorithm~\cite{clarkson2009numerical} is due to Clarkson and Woodruff.  All bounds assume each matrix entry requires $O(\log nd)$ bits.  It is randomized and it constructs a decomposition of a rank $k$ matrix $\hat A$ that satisfies $\|A - \hat A\|_F \leq (1+\eps) \|A - A_k\|_F$, with probability at least $1-\delta$. This provides a relative error construction bound of size $O((k/\eps)(n + d) \log(nd))$ bits.  
They also show an $\Omega((k/\eps)(n+d))$ bits lower bound.  

Although not explicitly described in their paper, one can directly use their techniques and analysis to achieve a weak form of a projection bound.  One maintains a matrix $B = A S$ with $m = O((k/\eps) \log(1/\delta))$ columns where $S$ is a $d \times m$ matrix where each entry is chosen from $\{-1,+1\}$ at random.
Then setting $\hat A = \pi_{B}(A)$, achieves a $(1+\eps)$ projection bound, however $B$ is rank $O((k/\eps) \log(1/\delta))$ and hence that is the only bound on $\hat A$ as well.  
The construction lower bound suggests that there is an $\Omega(dk/\eps)$ bits lower bound for projection, but this is not directly proven. 

They also study this problem in the \emph{turnstile} model where each element of the matrix can be updated at each step (including subtractions).  In this model they require $O((k/\eps^2)(n+d/\eps^2)\log(nd))$ bits, and show an $\Omega((k/\eps)(n+d)\log(nd))$ bits lower bound.  

Another more general ``coreset" result is provided by Feldman \etal~\cite{FSS13}.  In the streaming setting it requires $O((k/\eps) \log n)$ space and can be shown to provide a rank $O(k/\eps)$ matrix $B$ that satisfies a relative error bound of the form $\|A - \pi_B(A)\|_F^2 \leq (1+\eps) \|A - A_k\|_F^2$.  

\paragraph{Column sampling.}
Another line of work~\cite{drineas2003pass, drineas2006fast3, drineas2008relative, mahoney2009cur,boutsidis2011near,deshpande2006adaptive,drineas2006fast1, drineas2006fast2, rudelson2007sampling,achlioptas2001fast} considers selecting a set of rows from $A$ directly (not maintaining rows that for instance may be linear combinations of rows of $A$).  This maintains sparsity of $A$ implicitly and the resulting summary may be more easily interpretable.  Note, they typically consider the transpose of our problem and select columns instead of rows, and sometimes both.  
An algorithm~\cite{boutsidis2011near} can construct a set of $\ell = (2k/\eps)(1 + o(1))$ columns $R$ so that
$\|A - \pi_R(A)\|_F^2 \leq (1+\eps) \|A - A_k\|_F^2$.  There is an $\Omega(k/\eps)$ lower bound~\cite{deshpande2006adaptive}, but this enforces that only rows of the original matrix are retained and does not directly apply to our problem.  And these are not streaming algorithms.  

Although not typically described as streaming algorithms (perhaps because the focus was on sampling columns which already have length $n$) when a matrix is processed row wise there exists algorithms that can use reservoir sampling to become streaming.  
The best streaming algorithm~\cite{drineas2006fast2} samples $O(k /\eps^2)$ rows (proportional to their squared norm) to obtain a matrix $R$ so that
$\|A - \pi_R(A)\|_F^2 \leq \|A - A_k\|_F^2 + \eps \|A\|_F^2$, a weaker additive error bound.  
These techniques can also build approximate decompositions of $\hat A$ instead of using $\pi_R(A)$, but again these decompositions are only known to work with at least $2$ passes, and are thus not streaming.    

\paragraph{Other.}
There is a wealth of literature on this problem; most recently two algorithms~\cite{clarkson2013low,NN13} showed how to construct a decomposition of $\hat A$ that has rank $k$ with error bound 
$\|A - \hat A\|_F^2 \leq (1+\eps) \|A - A_k\|_F^2$
with constant probability in approximately $O(\nnz(A))$ time.  We refer to these papers for a more thorough survey of the history of the area, many other results, and other similar approximate linear algebra applications.  
But we attempt to report many of the most important related results in Appendix \ref{app:tables}.

Finally we mention a recent algorithm by Liberty~\cite{Lib12} which runs in $O(nd/\eps)$ time, maintains a matrix with $2/\eps$ rows in a row-wise streaming algorithm, and produces a matrix $\hat A$ of rank at most $2/\eps$ so that for any unit vector $x$ of length $d$ satisfies
$0 \leq \|A x\|^2 - \|\hat A x\|^2 \leq \eps \|A\|^2_F$.  We examine a slight variation of this algorithm and describe bounds that it achieves in more familiar terms.  

\paragraph{Incremental PCA.}
We mention one additional line of work on \emph{incremental PCA}~\cite{golub2012matrix, hall1998incremental, levey2000sequential, brand2002incremental,ross2008incremental}.  
These approaches attempt to maintain the PCA of a dataset $A$ (using SVD and a constant amount of additional bookkeeping space) as each row of $A$ arrives in a stream.  In particular, after $i-1$ rows they consider maintaining $A^i_k$, and on a new row $a_i$ compute $\svd([A^i_k; a_i]) = U^i S^i (V^i)^T$ and, then only retain its top rank $k$ approximation as $A^{i+1}_k = U^i_k S^i_k (V^i_k)^T$.  
This is remarkably similar to Liberty's algorithm~\cite{Lib12}, but is missing the Misra-Gries~\cite{mg-fre-82} step (we describe Liberty's algorithm in more detail in Section \ref{sec:FreqDir}).  As a result, incremental PCA can have arbitrarily bad error on adversarial data.  

Consider an example where the first $k$ rows generate a matrix $A_k$ with $k$th singular value $\sigma_k = 10$.  Then each row thereafter $a_i$ for $i>k$ is orthogonal to the first $k$ rows of $A$, and has norm $5$.  This will cause the $(k+1)$th right singular vector and value $\sigma_{k+1}$ of $\svd([A^i_k; a_i])$ to exactly describe the subspace of $a_i$ with $\sigma_{k+1} = 5$.  Thus this row $a_i$ will always be removed on the processing step and $A^{i+1}_k$ will be unchanged from $A^i_k$.  If all rows $a_i$ for $i>k$ are pointing in the same direction, this can cause arbitrarily bad errors of all forms of measuring approximation error considered above.  

\subsection{Our Results}

Our main result is a deterministic relative error bound for low-rank matrix approximation.  A major highlight is that all proofs are, we believe, quite easy to follow. 

\paragraph{Low-rank matrix approximation.}
We slightly adapt the streaming algorithm of Liberty~\cite{Lib12}, called \emph{Frequent Directions} to maintain $\ell = \lceil k + k/\eps\rceil$ rows, which outputs an $\ell \times d$ matrix $Q$.  Then we consider $Q_k$ a $k \times d$ matrix, the best rank $k$ approximation to $Q$ (which turns out to be its top $k$ rows).  
We show that 
\[
\|A - \pi_{Q_k}(A)\|^2_F \leq (1+\eps) \|A - A_k\|^2_F
\]
and that
\[
\|A - A_k\|_F^2 \leq \|A\|_F^2 - \|Q_k\|_F^2 \leq (1+\eps) \|A - A_k\|_F^2.
\]
This algorithm runs in $O(ndk^2/\eps^2)$ time.  If we allow $\ell = c \lceil k + k/\eps \rceil$ for any constant $c > 1$, then it can be made to run in $O(ndk/\eps)$ time with the same guarantees on $Q_k$.  

This is the smallest space streaming algorithm known for these bounds.  Also, it is deterministic, whereas previous algorithms were randomized.  

We note that it is common for the bounds to be written without squared norms, for instance as
$\|A - \pi_{Q_k}(Q)\|_F \leq (1+\eps) \|A - A_k\|_F$.  
For $\eps > 0$, if we take the square root of both sides of the bound above 
$\|A - \pi_{Q_k}(A)\|_F^2 \leq (1+\eps) \|A - A_k\|_F^2$, 
then we still get a $\sqrt{(1+\eps)} \leq (1+\eps)$ approximation.  

\paragraph{No sparse Frequent Directions.}
We also consider trying to adapt the Frequent Directions algorithm to column sampling (or rather row sampling), in a way that the $\ell$ rows it maintains are rows from the original matrix $A$ (possibly re-weighted).  This would implicitly preserve the sparsity of $A$ in $Q$.  
We show that this is, unfortunately, not possible.  

\subsection{Matrix Notation}
Here we quickly review some notation.  An $n \times d$ matrix $A$ can be written as a set of $n$ rows as $[a_1; a_2; \ldots, a_n]$ where each $a_i$ is a row of length $d$.  Alternatively a matrix $V$ can be written as a set of columns $[v_1, v_2, \ldots, v_d]$.  

The Frobenius norm of a matrix $A$ is defined $\|A\|_F = \sqrt{\sum_{i=1} \|a_i\|^2}$ where $\|a_i\|$ is Euclidean norm of $a_i$.  
Let $A_k$ be the best rank $k$ approximation of the matrix $A$, specifically $A_k = {\arg \max}_{C : \rank(C) \leq k} \|A - C\|_F$.  

Given a row $r$ and a matrix $X$ let $\pi_X(r)$ be a \emph{projection} operation of $r$ onto the subspace spanned by $X$.  In particular, we will project onto the row space of $X$, and this can be written as $\pi_X(r) = r X^T(X X^T)^{+} X$ where $Y^+$ indicates taking the Moore-Penrose psuedoinverse of $Y$.
But whether it projects to the row space or the column space will not matter since we will always use the operator inside of a Frobenius norm.  

This operator can be defined to project matrices $R$ as well, denoted as $\pi_X(R)$, where this can be thought of as projecting each row of the matrix $R$ individually.  

\section{Review of Related Algorithms}

We begin by reviewing two streaming algorithms that our results can be seen as an extension.  The first is an algorithm for heavy-hitters from Misra-Gries~\cite{mg-fre-82} and its improved analysis by Berinde \etal~\cite{BCIS09}.  We re-prove the relevant part of these results in perhaps a simpler way.  
Next we describe the algorithm of Liberty~\cite{Lib12} for low-rank matrix approximation that our analysis is based on.  We again re-prove his result, with a few additional intermediate results we will need for our extended analysis.  
One familiar with the work of Misra-Gries~\cite{mg-fre-82}, Berinde \etal~\cite{BCIS09}, and Liberty~\cite{Lib12} can skip this section, although we will refer to some lemmas re-proven below.  

\subsection{Relative Error Heavy-Hitters}
\label{sec:RelMG}
Let $A = \{a_1, \ldots, a_n\}$ be a set of $n$ elements where each $a_i \in [u]$.  
Let $f_j = |\{a_i \in A \mid a_i = j\}|$ for $j \in [u]$.  
Assume without loss of generality that $f_j \geq f_{j+1}$ for all $j$, and define $F_k = \sum_{j=1}^k f_j$.  This is just for notation, and \emph{not} known ahead of time by algorithms.    

The Misra-Gries algorithm~\cite{mg-fre-82} finds counts $\hat f_j$ so that for all $j \in [u]$ we have $0 \leq f_j - \hat f_j \leq n/\ell$.  It only uses $\ell$ counters and $\ell$ associated labels and works in a streaming manner as follows, starting with all counters empty (i.e. a count of $0$).  
It processes each $a_i$ in (arbitrary) order.  
\vspace{-.05in}
\begin{itemize} \denselist
\item If $a_i$ matches a label, increment the associated counter.  
\item If not, and there is an empty counter, change the label of the counter to $a_i$ and set its counter to $1$.  
\item Otherwise, if there are no empty counters, then decrement all counters by $1$.  
\end{itemize}
\vspace{-.05in}
To return $\hat f_j$, if there is a label with $j$, then return the associated counter; otherwise return $0$.  

Let $r$ be the total number of times that all counters are decremented.  We can see that $r<n/\ell$ since each time one counter is decremented then all $\ell$ counters (plus the new element) are decremented and must have been non-empty before hand.  Thus this can occur at most $n/\ell$ times otherwise we would have decremented more counts than elements.  
This also implies that $f_j - \hat f_j \leq r < n/\ell$ since we only do not count an element if it is removed by one of $r$ decrements.  
This simple, clever algorithm, and its variants, have been rediscovered several times ~\cite{demaine2002frequency, karp2003simple, golab2003identifying,metwally2006integrated}.

Define $\hat F_k = \sum_{j=1}^k \hat f_j$ and let $R_k = \sum_{j=k+1}^u f_j = n - F_k$.  
The value $R_k$ represents the total counts that cannot be described (even optimally) if we only use $k$ counters.  A bound on $F_k - \hat F_k$ in terms of $R_k$ is more interesting than one in terms of $n$, since this algorithm is only useful when there are only really $k$ items that matter and the rest can be ignored.  We next reprove a result of Berinde \etal~\cite{BCIS09} (in their Appendix A).  

\begin{lemma}[Berinde \etal \cite{BCIS09}]
The number of decrements is at most $r \leq R_k /( \ell - k)$.
\label{lem:Tbound}
\end{lemma}
\begin{proof}
On each of $r$ decrements at least $\ell-k$ counters not in the top $k$ are decremented.  These decrements must come from $R_k$, so each can be charged to at least one count in $R_k$; the inequality follows.  
\end{proof}

\begin{theorem}
When using $\ell = \lceil k + k/\eps \rceil$ in the Misra-Gries algorithm $F_k - \hat F_k \leq \eps R_k$ and $f_j - \hat f_j \leq \frac{\eps}{k} R_k$.  
If we use $\ell = \lceil k + 1/\eps \rceil$, then $f_j - \hat f_j \leq \eps R_k$.
\end{theorem}
\begin{proof}
Using Lemma \ref{lem:Tbound} we have $r \leq R_k/( \ell - k)$.  
Since for all $j$ we have $f_j - \hat f_j \leq r$, then $F_k - \hat F_k \leq rk \leq R_k \frac{k}{\ell - k}$.  
Finally, setting $\ell = k + k/\eps$ results in $F_k - \hat F_k \leq \eps R_k$ and $r \leq \frac{R_k}{\ell - k} = \frac{\eps}{k} R_k$.  

Setting $\ell = k + 1/\eps$ results in $f_j - \hat f_j \leq r \leq \frac{R_k}{\ell-k} = \eps R_k$ for any $j$.  
\end{proof}

This result can be viewed as a warm up for the rank $k$ matrix approximation to come, as those techniques will follow a very similar strategy.

\subsection{Additive Error Frequent Directions}
\label{sec:FreqDir}

Recently Liberty~\cite{Lib12} discovered how to apply this technique towards sketching matrices.  Next we review his approach, and for perspective and completeness re-prove his main results.  

\paragraph{Algorithm.}
The input to the problem is an $n \times d$ matrix $A$ that has $n$ rows and $d$ columns.  It is sometimes convenient to think of each row $a_i$ as a point in $\b{R}^d$.  
We now process $A$ one row at a time in a streaming fashion always maintaining an $\ell \times d$ matrix such that for any unit vector $x \in \b{R}^d$ 
\begin{equation}
\|A x\|^2 - \|Q x\|^2 \leq \|A\|_F^2 / \ell,
\label{eq:FD-inv}
\end{equation}
This invariant (\ref{eq:FD-inv}) guarantees that in any ``direction'' $x$ (since $x$ is a unit vector in $\b{R}^d$), that $A$ and $Q$ are close, where close is defined by the Frobenius norm of $\|A\|_F^2$ over $\ell$.  

Liberty's algorithm is described in Algorithm \ref{alg:freqDir}.  
At the start of each round, the last row of $Q$ will be all zeros.  
To process each row $a_i$, 
we replace the last row (the $\ell$th row) of $Q$ with $a_i$ to create a matrix $Q_i$.  
We take the SVD of $Q_i$ as $[U,S,V] = \svd(Q_i)$.  
Let $\delta = s_\ell^2$, the last (and smallest) diagonal value of $S$, and in general let $s_j$ be the $j$th diagonal value so $S = \diag(s_1, s_2, \ldots, s_\ell)$.  
Now set $s'_j = \sqrt{s_j^2 - \delta}$ for $j \in [\ell]$, and notice that all values are non-negative and $s'_{\ell} = 0$.  
Set $S' = \diag(s'_1, s'_2, \ldots, s'_\ell)$.  
Finally set $Q = S' V^T$.

\begin{algorithm}
\caption{\label{alg:freqDir} Frequent Directions (Liberty \cite{Lib12})}
\begin{algorithmic}
\STATE  Initialize $Q^0$ as an all zeros $\ell \times d$ matrix. 
\FOR {each row $a_i \in A$}
  \STATE  Set $Q_+$ $\leftarrow$ $Q^{i-1}$ with last row replaced by $a_i$
  \STATE  $[Z, S, Y] = \svd(Q_+)$
  \STATE  $C^i = S Y^T$   \hspace{.3in} [\emph{only for notation}]
  \STATE  Set $\delta_i = s_{\ell}^2$  \hspace{.25in} [\emph{the $\ell$th entry of $S$, squared}]  
  \STATE  Set $S' = \diag\left(\sqrt{s_1^2 - \delta_i}, \sqrt{s_2^2 - \delta_i}, \ldots, \sqrt{s_{\ell-1}^2 - \delta_i}, 0\right)$.  
  \STATE  Set $Q^i = S' Y^T$  
\ENDFOR
\STATE \textbf{return} $Q = Q^n$
\end{algorithmic}
\end{algorithm}

It is useful to interpret each row of $Y^T$ as a ``direction,'' where the first row is along the direction with the most variance, all rows are orthogonal, and all rows are sorted in order of variance given that they are orthogonal to previous rows.  Then multiplying by $S'$ scales the $j$th row $y_j$ of $Y^T$ by $s'_j$.  Since $s'_\ell = 0$, then the last row of $Q^i$ must be zero.  

\paragraph{Analysis.}
Let $\Delta = \sum_{i=1}^n \delta_i$.  

\begin{lemma}
For any unit vector $x \in \b{R}^d$ we have $\|C^ix\|^2 - \|Q^ix\|^2 \leq \delta_i$.
\label{lem:L2normInIteration}
\end{lemma}
\begin{proof}
Let $Y_j$ be the $j$th column of $Y$, then
\[
\|C^i x\|^2 
= 
\sum_{j=1}^\ell s_j^2 \langle y_j, x\rangle^2
=
\sum_{j=1}^\ell ((s'_j)^2 + \delta_i) \langle y_j, x\rangle^2
= 
\sum_{j=1}^\ell (s'_j)^2 \langle y_j, x\rangle^2 + \delta_i \sum_{j=1}^\ell \langle y_j, x \rangle^2
\leq 
\|Q^i x\|^2 + \delta_i.
\]
Subtracting $\|Q^i x\|^2$ from both sides completes the proof.
\end{proof}

\begin{lemma}
For any unit vector $x \in \b{R}^d$ we have $0 \leq \|Ax\|^2 - \|Qx\|^2 \leq \Delta$.
\label{lem:L2normInFinal}
\end{lemma}
\begin{proof}
Notice that $\|C^i x \|^2 = \|Q^{i-1} x \|^2 + \|a_i x\|^2$ for all $2 \leq i \leq n$ and that $\|Q^1 x\|^2 = \|a_1 x\|^2$. By substituting this into inequality from Lemma \ref{lem:L2normInIteration}, we get
\[
\|Q^{i-1} x \|^2 + \|a_i x\|^2 \leq \|Q^i x\|^2 + \delta_i
\]
Subtracting $\|Q^{i-1} x\|^2$ from both sides and summing over $i$ reveals
\[
\|Ax\|^2 = \sum_{i=1}^n \|a_i x\|^2 \leq \sum_{i=1}^n (\|Q^i x\|^2 - \|Q^{i-1} x\|^2 + \delta_i)
= \|Q^n x\|^2 - \|Q^0 x\|^2 + \sum_{i=1}^n \delta_i = \|Q^n x\|^2 + \Delta.
\]
Subtracting $\|Q^n x\|^2 = \|Q x\|^2$ from both sides proves the second inequality of the lemma.

To see the first inequality observe
$\|Q^{i-1} x\|^2 + \|a_i x\|^2 = \|C^i x\|^2 \leq \|Q^i x\|^2$ for all $1\leq i \leq n$.  
Then we can expand  
\[
\|A x\|^2 = \sum_{i=1}^n \|a_i x\|^2 = \sum_{i=1}^n \|C^i x\|^2 - \|Q^{i-1} x\|^2
\geq
\sum_{i=1}^n \|Q^i x\|^2 - \|Q^{i-1} x\|^2 = \|Q x\|^2. \qedhere
\]
\end{proof}

\begin{lemma}[Liberty~\cite{Lib12}]\label{lem:fD}
Algorithm \ref{alg:freqDir} maintains for any unit vector $x$ that
\[
0 \leq \|Ax\|^2 - \|Qx\|^2 \leq \|A\|_F^2 / \ell
\]
and
\[
T = \Delta \ell = \|A\|_F^2 - \|Q\|_F^2.
\]
\end{lemma}
\begin{proof}
In the $i$th round of the algorithm $\|C^i\|^2_F = \|Q^i\|^2_F + \ell \delta_i$ and $\|C^i \|^2_F = \|Q^{i-1} \|^2_F + \|a_i\|^2$.  By solving for $\|a_i\|^2$ and summing over $i$ we get
\[
\|A\|^2_F = \sum_{i=1}^n \|a_i\|^2 = \sum_{i=1}^n \|Q^i\|^2_F - \|Q^{i-1}\|_F^2 + \ell \delta_i = 
\|Q\|^2_F + \ell \Delta.
\]
This proves the second part of the lemma.  Using that $\|Q\|_F^2 \geq 0$ we obtain $\Delta \leq \|A\|_F^2 /\ell$.  Substituting this into Lemma \ref{lem:L2normInFinal} yields $0 \leq \|A x\|^2 - \|Q x\|^2 \leq \|A\|_F^2 / \ell$.    
\end{proof}

\section{New Relative Error Bounds for Frequent Directions}
\label{sec:relative}
We now generalize the relative error type bounds for Misra-Gries (in Section \ref{sec:RelMG}) to the Frequent Directions algorithm (in Section \ref{sec:FreqDir}).  

Before we proceed with the analysis of the algorithm, we specify some parameters and slightly modify Algorithm \ref{alg:freqDir}.  We always set $\ell = \lceil k + k/\eps \rceil$.  Also, instead of returning $Q$ in Algorithm \ref{alg:freqDir}, as described by Liberty, we return $Q_k$.  Here $Q_k$ is the best rank $k$ approximation of $Q$ and can be written $Q_k = S_k' Y_k^T$ where $S'_k$ and $Y_k$ are the first $k$ rows of $S'$ and $Y$, respectively.
Note that $Y = [y_1, \ldots, y_\ell]$ are the right singular vectors of $Q$.

This way $Q_k$ is also rank $k$ (and size $k \times d$), and will have nice approximation properties to $A_k$.  
Recall that $A_k = U_k \Sigma_k V_k^T$ where $[U,\Sigma,V] = \svd(A)$ and $U_k$, $\Sigma_k$, $V_k$ are the first $k$ columns of these matrices, representing the first $k$ principal directions.  
Let $V = [v_1, \ldots, v_d]$ be the right singular vectors of $A$.  
\begin{lemma}
$\Delta \leq \|A - A_k\|_F^2 / (\ell -k)$.  
\label{lem:Dbound}
\end{lemma}
\begin{proof}
Recall that $T = \Delta \ell = \|A\|_F^2 - \|Q\|_F^2$ is the total squared norm subtracted from all of any set of orthogonal directions throughout the algorithm.
Now if $r = \rank(A)$ we have:
\begin{align*}
T 
&= 
\|A\|_F^2 - \|Q\|_F^2 &\text{By Lemma \ref{lem:fD}}
\\ &= 
\sum_{i=1}^k \|A v_i\|^2 + \sum_{i=k+1}^r \|A v_i\|^2 - \|Q\|_F^2 & 
\\&= 
\sum_{i=1}^k \|A v_i\|^2 + \|A - A_k \|_F^2 - \|Q\|_F^2 &
\\ &\leq 
\sum_{i=1}^k \|A v_i\|^2 + \|A - A_k \|_F^2 - \sum_{i=1}^k \|Q v_i\|^2 & \text{$\sum_{i=1}^k \|Q v_i\|^2 < \|Q\|_F^2$}
\\&= 
\|A - A_k \|_F^2 + \sum_{i=1}^k \left (\|A v_i\|^2 - \|Q v_i\|^2 \right) & 
\\&\leq 
\|A - A_k \|_F^2 + k\Delta & \text{By Lemma \ref{lem:L2normInFinal} $\|A v_i\|^2 - \|Q v_i\|^2 \leq \Delta$}
\end{align*}  

Now we solve for $T = \Delta \ell \leq \|A - A_k\|_F^2 + k \Delta$ to get
$\Delta \leq \|A - A_k\|_F^2 / (\ell - k)$.  
\end{proof}

Now we can show that projecting $A$ onto $Q_k$ provides a relative error approximation.  

\begin{lemma} 
$\|A - \pi_{Q_k}(A)\|_F^2 \leq (1+\eps) \|A - A_k\|_F^2.$
\end{lemma}
\begin{proof}
Using the vectors $v_i$ as right singular vectors of $A$, and letting $r = \rank(A)$, then we have
\begin{align*}
\|A - \pi_{Q_k}(A)\|_F^2 
&= 
\|A\|_F^2 - \|\pi_{Q_k}(A)\|_F^2 = 
\|A\|_F^2 - \sum_{i=1}^k \|A y_i\|^2                                                  & \text{Pythagorean theorem}
\\& \leq 
\|A\|_F^2 - \sum_{i=1}^k\|Q y_i\|_F^2 
& \text{By Lemma \ref{lem:L2normInFinal}}
\\& \leq 
\|A\|_F^2 - \sum_{i=1}^k \|Q v_i\|^2    & \text{Since $\sum_{i=1}^j \|Q y_i\|^2 \geq \sum_{i=1}^j \|Q v_i\|^2$}
\\ &\leq
\|A\|_F^2 - \sum_{i=1}^k (\|A v_i\|^2 - \Delta)  
                                                              & \text{By Lemma \ref{lem:L2normInFinal}}
\\ & = \|A\|_F^2 - \|A_k\|_F^2 + k\Delta &
\\ &\leq
\|A - A_k\|_F^2 + \frac{k}{\ell-k} \|A - A_k\|_F^2 = \frac{\ell}{\ell-k} \|A - A_k\|_F^2  
                                                              & \text{By Lemma \ref{lem:Dbound}} 
\end{align*}
Finally, setting $\ell = \lceil k + k/\eps \rceil$ results in $\|A - \pi_{Q_k}(A)\|_F^2 \leq (1+\eps) \|A - A_k\|_F^2$. 
\end{proof}

We would also like to relate the Frobenius norm of $Q_k$ directly to that of $A_k$, instead of projecting $A$ onto it (which cannot be done in a streaming setting, at least not in $\omega(n)$ space).  However $\|A - Q_k\|_F$ does not make sense since $Q_k$ has a different number of rows than $A$.  
However, we can decompose $\|A - A_k\|_F^2 = \|A\|_F^2 - \|A_k\|_F^2$ since $A_k$ is a projection of $A$ onto a (the best rank $k$) subspace, and we can use the Pythagorean Theorem.  Now we can compare $\|A\|_F^2 - \|A_k\|_F^2$ to $\|A\|_F^2 - \|Q_k\|_F^2$.  

\begin{lemma} \label{lem:QktoAk}
$\|A\|_F^2 - \|A_k\|_F^2 \leq \|A\|_F^2 - \|Q_k\|_F^2 \leq (1+\eps) (\|A\|_F^2 - \|A_k\|_F^2)$.
\end{lemma}
\begin{proof}
The first inequality can be seen since
\[
\|A_k\|_F^2 = \sum_{i=1}^k \|A v_i\|^2 \geq \sum_{i=1}^k \|A y_i\|^2 \geq \sum_{i=1}^k \|Q y_i\|^2 = \|Q_k\|_F^2.
\]
And the second inequality follows by
\begin{align*}
\|A\|_F^2 - \|Q_k\|_F^2
&= 
\|A\|_F^2 - \sum_{i=1}^k \|Q y_i\|^2
\\ & \leq
\|A\|_F^2 - \sum_{i=1}^k \|Q v_i\|^2
\\ & \leq
\|A\|_F^2 - \sum_{i=1}^k (\|A v_i\|^2 - \Delta) = \|A\|_F^2 - \|A_k\|_F^2 + k \Delta
\\ & \leq
\|A - A_k\|_F^2 + \frac{k}{\ell-k} \|A - A_k\|_F^2 = \frac{\ell}{\ell-k} \|A - A_k\|_F^2 
\end{align*}
Finally, setting $\ell = k + k/\eps$ results in $\|A\|_F^2 - \|Q_k\|_F^2 \leq (1+\eps) \|A - A_k\|_F^2 = (1+\eps)(\|A\|_F^2 - \|A_k\|_F^2)$. 
\end{proof}

One may ask why not compare $\|A_k\|_F$ to $\|Q_k\|_F$ directly, instead of subtracting from $\|A\|_F^2$.  
First note that the above bound \emph{does} guarantee that $\|A_k\|_F \geq \|Q_k\|_F$.  
Second, in situations where a rank $k$ approximation is interesting, then most of the mass from $A$ should be in its top $k$ components.  Then $\|A_k\|_F>\|A - A_k\|_F$ so the above bound is actually tighter.  To demonstrate this we can state the following conditional statement comparing $\|A_k\|_F$ and $\|Q_k\|_F$.  

\begin{lemma}
If $\|A-A_k\|_F \leq \|A_k\|_F$, then 
\[
(1-\eps)\|A_k\|_F^2 \leq \|Q_k\|_F^2 \leq \|A_k\|_F^2.
\]
\end{lemma}
\begin{proof}

The second inequality follows from Lemma \ref{lem:QktoAk}, by subtracting $\|A\|_F^2$.  
The first inequality uses Lemma \ref{lem:L2normInFinal}.  
\begin{align*}
\|A_k\|_F^2 &= \sum_{i=1}^k \|A v_i\|^2 
\leq
\sum_{i=1}^k (\|Q v_i\|^2 + \Delta) 
\leq
\|Q_k\|_F^2 + k\Delta
\\ &\leq
\|Q_k\|_F^2 + \frac{k}{\ell-k} \|A - A_k\|_F^2 
\leq \|Q_k\|_F^2 + \eps \|A_k\|_F^2.  \qedhere
\end{align*}
\end{proof}

Finally, we summarize all of our bounds about Algorithm \ref{alg:freqDir}.  

\begin{theorem}
Given an input $n \times d$ matrix $A$, by setting $\ell = \lceil k + k/\eps \rceil$
Algorithm \ref{alg:freqDir} runs in time $O(n d \ell^2) = O(ndk^2/\eps^2)$ time and produces an $\ell \times d$ matrix $Q$ that for any unit vector $x \in \b{R}^d$ satisfies
\[
0 \leq \|Ax\|^2 - \|Q x\|^2 \leq \|A\|_F^2/\ell
\]
and the projection of $Q$ along its top $k$ right singular values is a $k \times d$ matrix $Q_k$ which satisfies
\[
\|A - \pi_{Q_k}(A)\|_F^2  \leq (1+\eps)\|A - A_k\|_F^2
\]  
and
\[
\|A\|_F^2 - \|A_k\|_F^2 \leq \|A\|_F^2 - \|Q_k\|_F^2 \leq (1+\eps) (\|A\|_F^2 - \|A_k\|_F^2).
\]
\end{theorem}

Liberty~\cite{Lib12} also observes that by increasing $\ell$ by a constant $c > 1$ and then processing every $\ell (c - 1)$ elements in a batch setting (each round results in a $c \ell$ row matrix $Q$) then the runtime can be reduced to $O(\frac{c^2}{c-1} n d \ell) = O(nd k/\eps)$ at the expense of more space.  The same trick can be applied here to use $\ell = c \lceil k + 2k/\eps \rceil$ rows in total $O(ndk/\eps)$ time.

\section{No Sparse Frequent Directions}
\label{sec:Sparse Frequent Directions}
In this section we consider extending the frequent directions algorithm described in the previous section to a sparse version.  The specific goal is to retain a (re-weighted) set of $\ell$ rows $Q$ of an input matrix $A$ so that for any unit vector $x \in \b{R}^d$ that
$0 \leq \|Ax\|^2 - \|Qx\|^2 \leq \|A\|_F^2 /\ell$, and also to hopefully extend this so that
$\|A - \pi_{Q_k}(A)\|_F^2 \leq (1+\eps)\|A - A_k\|_F^2$ as above.  
This is an open problem left by Liberty~\cite{Lib12}.  It is also a useful goal in many scenarios when returning a set of $k$ singular vectors of a matrix that are linear combinations of inputs are hard to interpret; in this case returning a weighted set of actual rows is much more informative.  

In this section we show that this is not possible by extending the frequent directions algorithm.  

In particular we consider processing one row in the above framework.  The input to the problem is an $\ell \times d$ matrix $Q = [w_1 \bar r_1; \ldots; w_\ell \bar r_\ell]$ where each $w_j$ is a scalar (initially set to $w_j = \|r_j\|$).  
The output of one round should be an ${\ell-1} \times d$ matrix $\hat Q = [\hat w_1 \bar r_1; \ldots ; \hat w_{j-1} \bar r_{j-1}; \hat w_{j+1} \bar r_{j+1}; \ldots; \hat w_{\ell} \bar r_{\ell}]$ where one of the rows, namely $r_j$, is removed and the rest of the rows are re-weighted.  

\paragraph{Requirements.}
To make this process work we need the following requirements.  
Let $\delta_i = \min_j \|\op_{Q_{-j}}(Q)\|^2$ represent the smallest amount of squared norm resulting from removing one row from $Q$ by the procedure above.  
\footnote{Define $\op_X(Y)$ as the \emph{orthogonal projection} of $Y$ onto $X$.  It projects each row of $Y$ onto the subspace orthogonal to the basis of $X$.  It can be interpreted as $\op_X(Y) = Y - \pi_X(Y)$.  Also, we let $Q_{-j}$ be the matrix $Q$ after removing the $j$th row.}
Assume we remove this row, although removing any row is just as difficult but would create even more error.    
\begin{itemize} \denselist
 \item[(P1)] 
The Frobenius norm $\|\hat Q\|_F^2$ must be reduced so 
$
\|\hat Q\|_F^2 \leq \|Q\|_F^2 - c \ell \delta_i 
$
for some absolute constant $c$.  
\\
The larger the constant (ideally $c=1$) the smaller the bound on $\ell$.  
 \item[(P2)] 
For any unit vector (direction) $x \in \b{R}^d$ the difference in norms between $\hat Q$ and $Q$ must be bounded as
$
\|Q x\|^2 \leq \|\hat Q x \|^2 + \delta_i.
$
\\
In the direction $v$ which defines the norm $\delta_i = \|\op_{Q_{-j}}(Q)\|^2 = \|Q v\|^2$ we have $0 = \|\hat Q v\|^2$ and the inequality is tight.  And for instance a vector $u$ in the span of $Q_{-j}$ we have that $\|Qu\|^2 = \|\hat Q u\|^2$, which makes the right hand side larger.  
\end{itemize}

If both (P1) and (P2) hold, then we can run this procedure for all rows of $A$ and obtain a final matrix $Q$.  Using similar analysis as in Section \ref{sec:FreqDir}, for any unit vector $x \in \b{R}^d$ we can show
\[
\|Ax\|^2 - \|Q x\|^2 \leq \sum_i \delta_i \leq \|A\|^2_F / (c \ell).
\]

\paragraph{Hard construction.}
Consider a $\ell \times d$ matrix $Q$ with $d > \ell$.  
Let each row of $Q$ be of the form $r_j = [1, 0, 0, \ldots, 0, 1, 0, \ldots, 0]$ where there is always a $1$ in the first column, and another in the $(j+1)$th column for row $j$; the remaining entries are $0$.  
Let $x = [1, 0, 0, \ldots, 0]$ be the direction strictly along the dimension represented by the first column.  

Now $\delta_i = \min_j \|\op_{Q_{-j}}(Q)\|^2 = 1$, since for any $j$th row $r_j$, when doing an orthogonal projection to $Q_{-j}$ the remaining vector is always exactly in the $j$th column where that row has a squared norm of $1$.  
For notational simplicity, lets assume we choose to remove row $\ell$.  

We now must re-weight rows $r_1$ though $r_{\ell-1}$; let the new weights be $\hat w_j^2 = w_j^2 - \alpha_j$.  

In order to satisfy (P1) we must have 
\[
\|Q\|_F^2 - \|\hat Q\|_F^2 
= 
\sum_{j=1}^{\ell} w_j^2 - \sum_{j=1}^{\ell-1} \hat w_j^2
=
w_{\ell}^2 + \sum_{j=1}^{\ell-1} \alpha_j
\geq
c \ell \delta_i.
\]
Since $w_{\ell}^2 = 2$ and $\delta_i = 1$ we must have 
$\sum_{j=1}^{\ell-1} \alpha_j \geq c \ell-2$.  

In order to satisfy (P2) we consider the vector $x$ as defined above.  
We can observe
\[
\|Q x\|^2 = \sum_{j=1}^{\ell} w_j^2 \langle \bar r_j, x \rangle^2 = \sum_{j=1}^{\ell} w_j^2 (1/2).
\]
and
\[
\|\hat Q x\|^2 = \sum_{j=1}^{\ell-1} \hat w_j^2 \langle \bar r_j, x\rangle^2 = (1/2) \sum_{j=1}^{\ell-1} (w_j^2 - \alpha_j) = 
(\|Q x\|^2 - 1) - (1/2) \sum_{j=1}^{\ell-1} \alpha_j.
\]
Thus we require that $\sum_{j=1}^{\ell-1} \alpha_j \leq 2 \delta_i - 2 = 0$, since recall $\delta_i = 1$.  

Combining these requirements yields that 
\[
0 \geq \sum_{j=1}^{\ell-1} \alpha_j \geq c \ell-2
\]
which is only valid when $c \leq 2/\ell$.  

Applying the same proof technique as in Section \ref{sec:FreqDir} to this process reveals, at best, a bound so that for any direction $x \in \b{R}^d$ we have
\[
\|A x\|^2 - \|Q x\|^2 \leq \|A\|_F^2/2. 
\]

\subsection*{Acknowledgements}
We thank Edo Liberty for encouragement and helpful comments, including pointing out several mistakes in an earlier version of this paper.  And thank David P. Woodruff, Christos Boutsidis, Dan Feldman, and Christian Sohler for helping interpret some results.  

\bibliographystyle{plain}
\bibliography{mina,discrepancy}

\newpage

\appendix
\section{Tables of Previous Bounds}
\label{app:tables}

In this appendix we try to survey the landscape of work in low-rank matrix approximation.  There are many types of bounds, models of construction, and algorithms.  We tried to group them into three main categories: Streaming, Fast Runtime, and Column Sampling.  We also tried to write bounds in a consistent compatible format.  To do so, some parts needed to be slightly simplified -- hopefully we got everything correct.  The authors will be glad to know if we missed or misrepresented any results.  

The space and time bounds are given in terms of $n$ (the number of rows), $d$ (the number of columns), $k$ (the specified rank to approximate), $r$ (the rank of input matrix $A$), $\eps$ (an error parameter), and $\delta$ (the probability of failure of a randomized algorithm).  
An expresion $\tilde O(x)$ hides $\textsf{poly}\log (x)$ terms.

The size is sometimes measured in terms of the number of columns (\#C) and/or the number of rows (\#R).  Otherwise, if \#R or \#C is not specified the space refers the number of words in the RAM model where it is assumed $O(\log nd)$ bits fit in a single word.  

The error is of one of several forms.  
\begin{itemize}\denselist
\item A \emph{projection} result builds a subspace $G$ so that $\hat{A} = \pi_G(A)$, but does not actually construct $\pi_G(A)$.  This is denoted by \textsf{P}.  
\\
Ideally $\rank(G) = k$.  When that is not the case, then it is denoted $\textsf{P}_r$ where $r$ is replaced by the rank of $G$.  
\item A \emph{construction} result builds a series of (usually 3) matrices (say $C$, $U$, and $R$) where $\hat{A} = CUR$.  Note again, it does not construct $\hat{A}$ since it may be of larger size than all of $C$, $U$, and $R$ together, but the three matrices can be used in place of $\hat{A}$.   This is denoted \textsf{C}.
\item \emph{$\eps$-relative error} is of the form 
$\|A - \hat A\|_F \leq (1+\eps)\|A - A_k\|_F$ where $A_k$ is the best rank-$k$ approximation to $A$.  This is denoted $\eps$\textsf{R}.  
\item \emph{$\eps$-additive error} is of the form
$\|A- \hat A\|^2_F \leq \|A - A_k\|^2_F + \eps \|A\|^2_F$.  This is denoted $\eps$\textsf{A}.
\\
This can sometimes also be expressed as a spectral norm of the form
$\|A - \hat A\|^2_2 \leq \|A - A_k\|^2_2 + \eps \|A\|^2_F$ (note the error term $\eps \|A\|^2_F$ still has a Frobenius norm).  This is denoted $\eps$\textsf{L}$_2$.  
\item In a few cases the error does not follow these patterns and we specially denote it.  
\item Algorithms are randomized unless it is specified.  In all tables we state bounds for a constant probability of failure.  If we want to decrease the probability of failure to some parameter $\delta$, we can generally increase the size and runtime by $O(\log(1/\delta))$.  \end{itemize}

\begin{table}[b]
  \begin{tabular}{|p{2.7cm}|p{3.8cm}|p{4.5cm}|p{4cm}|}
\cline{1-4}  
\multicolumn{4}{|c|}{\label{tbl:steam}\textbf{Streaming algorithms}}   \\
\cline{1-4}
\textbf{Paper} & \textbf{Space} & \textbf{Time} & \textbf{Bound} \\
     \cline{1-4} 
 DKM06\cite{drineas2006fast2} \newline LinearTimeSVD  & 
    \#R = $O(1/\eps^2)$ \newline $O((d+1/\eps^2)/\eps^4)$ & 
    $O((d+1/\eps^2)/\eps^4 + \nnz(A))$ & 
    \pbox{4cm}{\textsf{P}, $\eps$\textsf{L}$_2$} \\
    \cline{2-4}
    & \#R = $O(k/\eps^2)$ \newline $O((k/\eps^2)^2(d+k/\eps^2))$ & 
    $O((k/\eps^2)^2(d+k/\eps^2)+ \nnz(A))$ & 
    \pbox{4cm}{\textsf{P}, $\eps$\textsf{A}} \\
    \cline{1-4}
  Sar06\cite{sarlos2006improved}   \newline turnstile
    & \#R  = $O(k/\eps + k\log k)$ \newline 
        $O(d(k/\eps + k\log k))$ & 
        $O(\nnz(A)(k/\eps+k \log k) + d(k/\eps+k \log k)^2))$ & 
       \textsf{P}$_{O(k/\eps + k \log k)}$, $\eps$\textsf{R}\\
 \cline{1-4} 
 	CW09\cite{clarkson2009numerical}   &\#R =  $O(k/\eps)$  & $O(nd^2 + (ndk/\eps))$ & \textsf{P}$_{O(k/\eps)}$, $\eps$\textsf{R} \\
 \cline{1-4} 
 	CW09\cite{clarkson2009numerical} &$O((n+d)(k/\eps))$  & $O(nd^2 + (ndk/\eps))$ & \textsf{C}, $\eps$\textsf{R} \\
 	\cline{1-4}
 	CW09\cite{clarkson2009numerical} \newline turnstile & $O((k/\eps^2) (n+d/\eps^{2}) )$  & $O(n(k/\eps^2)^2 + nd(k/\eps^2)+nd^2)$ & \textsf{C}, $\eps$\textsf{R} \\
 	\cline{1-4}
    FSS13\cite{FSS13} \newline deterministic & $O((dk/\eps) \log n)$ & $n ((dk/\eps)\log n)^{O(1)}$ &
    $\textsf{P}_{2\lceil k/\eps \rceil}$, $\eps \textsf{R}$ \\
 	\cline{1-4}
    Lib13\cite{Lib12} \newline deterministic, $\rank(Q) \leq 2/\eps$ & \#R = $2/\eps$ \newline $O(d/\eps)$ & $O(nd/\eps)$ & Any unit vector $x$ \newline $0 \leq \|Ax\|^2 - \|Qx\|^2 \leq \eps\|A\|_F^2$ \\
 	\cline{1-4}
    Lib13\cite{Lib12} \newline deterministic,  \newline $\rho = \|A\|_F^2/\|A\|_2^2$ & \#R = $O(\rho/\eps)$ \newline $O(d \rho /\eps)$ & $O(nd \rho /\eps)$ & \textsf{P}$_{O(\rho/\eps)}$, $\eps$\textsf{L}$_2$ \\
  \hline \hline
    \textit{This paper} \newline deterministic & \#R = $\lceil k/\eps + k \rceil$  \newline $O(dk/\eps)$ & $O(ndk^2/\eps^2)$ & \textsf{P}, $\eps$\textsf{R}\\
	\cline{2-4}    
    & \#R = $c \lceil k/\eps + k \rceil$, $c > 1$  \newline $O(dk/\eps)$ & $O(\frac{c^2}{c-1}ndk/\eps)$ & \textsf{P}, $\eps$\textsf{R}\\
	\cline{1-4}    
    \end{tabular}
\end{table}

\begin{table}
  \begin{tabular}{|p{2cm}|p{4.5cm}|p{4.7cm}|p{3.8cm}|} 

\cline{1-4}  
\multicolumn{4}{|c|}{\label{tbl:time}\textbf{Algorithms with Fast Runtime}}   \\
\cline{1-4}
\textbf{Paper} & \textbf{Space} & \textbf{Time} & \textbf{Bound} \\
 \cline{1-4} 
    AM01\cite{achlioptas2001fast} & $O(nd)$ & $O(nd)$ & $\|A-\hat{A}_k\|_2 \leq \|A-A_k\|_2 + 10(\max_{i,j} |A_{ij}|)\sqrt{n+d}$\\
    \cline{1-4}
    CW13\cite{clarkson2013low} & 
     $O((k^2/\eps^6)\log^4 (k/\eps)+(nk/\eps^3)\log^2(k/\eps)+(nk/\eps)\log(k/\eps))$ & 
     $O(\nnz(A)) + \tilde{O}(nk^2/\eps^4+k^3/\eps^5)$ & 
     \textsf{C}, $\eps$\textsf{R} \\
     \cline{1-4} 
     NN13\cite{NN13} &?& $O(\nnz(A))+\tilde{O}(nk^2+nk^{1.37}\eps^{-3.37}+k^2.37\eps^{-4.37})$ &  \textsf{C}, $\eps$\textsf{R} \\
     \cline{2-4} 
     &?& $O(\nnz(A) \log^{O(1)}{k}) + \tilde{O}(nk^{1.37}\eps^{-3.37}+k^{2.37}\eps^{-4.37})$ &  \textsf{C}, $\eps$\textsf{R} \\
     \cline{1-4} 
    \end{tabular}
\end{table}

\makeatletter
\setlength{\@fptop}{0pt}
\makeatother

\begin{table}
\label{tbl:ColSamp}
  \begin{tabular}{|p{3.3cm}|p{3.7cm}|p{5.5cm}|p{2.5cm}|}
 
\cline{1-4}  
\multicolumn{4}{|c|}{\label{tbl:column}\textbf{Column Sampling algorithms}}   \\
\cline{1-4}
\textbf{Paper} & \textbf{Space} & \textbf{Time} & \textbf{Bound} \\
    \cline{1-4} 
      FKV04\cite{frieze2004fast} & 
     $O(k^4/\eps^{6}\max(k^4,\eps^{-2}))$ & 
     $O(k^5/\eps^{6} \max(k^4,\eps^{-2}))$ & 
     \textsf{P}, $\eps$\textsf{A} \\
    \cline{1-4}
    DV06~\cite{deshpande2006adaptive}  & 
    \#C = $O(k/\eps+k^2\log k)$ \newline $O(n(k/\eps + k^2 \log k))$ & 
    $O(\nnz(A)(k/\eps+k^2\log k)+$ \newline $(n+d)(k^2/\eps^2+k^3\log (k/\eps)+k^4\log^2k))$ & 
    \pbox{4cm}{\textsf{P}, $\eps$\textsf{R}}\\ 
 \cline{1-4} 
 DKM06\cite{drineas2006fast2} \newline ``LinearTimeSVD''  & 
    \#C = $O(1/\eps^2)$ \newline $O((n+1/\eps^2)/\eps^4)$ & 
    $O((n+1/\eps^2)/\eps^4 + \nnz(A))$ & 
    \pbox{4cm}{\textsf{P}, $\eps$\textsf{L}$_2$} \\
    \cline{2-4}
    & \#C = $O(k/\eps^2)$ \newline $O((k/\eps^2)(n+k/\eps^2))$ & 
    $O((k/\eps^2)^2(n+k/\eps^2) + \nnz(A))$ & 
    \pbox{4cm}{\textsf{P}, $\eps$\textsf{A}} \\
    \cline{1-4}
    DKM06\cite{drineas2006fast2} \newline ``ConstantTimeSVD''  & 
    \#C+R = $O(1/\eps^4)$ \newline $O(1/\eps^{12} + nk/\eps^4)$ & 
    $O((1/\eps^{12}+nk/\eps^{4} + \nnz(A))$ & 
    \pbox{4cm}{\textsf{P}, $\eps$\textsf{L}$_2$} \\
    \cline{2-4}
    & \#C+R = $O(k^2/\eps^4)$ \newline $O(k^6/\eps^{12} + nk^3/\eps^4)$ & 
    $O(k^6/\eps^{12}+nk^3/\eps^4 + \nnz(A))$ & 
    \pbox{4cm}{\textsf{P}, $\eps$\textsf{A}} \\
  \cline{1-4}
   DMM08~\cite{drineas2008relative} \newline ``CUR'' & 
   \#C =$O(k^2/\eps^{2})$ \newline \#R = $O(k^4/\eps^{6})$ & 
   $O(nd^2) $ & 
   \pbox{4cm}{\textsf{C}, $\eps$\textsf{R}}\\ 
    \cline{1-4}
    MD09~\cite{mahoney2009cur} \newline ``ColumnSelect''  & 
    \#C = $O(k \log k/\eps^{2})$ \newline $O(nk \log k /\eps^{2})$ & 
    $O(nd^2)$ & 
    \pbox{4cm}{$\textsf{P}_{O(k \log k/\eps^{2})}$, $\eps$\textsf{R}}\\
    \cline{1-4}
   BDM11~\cite{boutsidis2011near}  & 
   \#C = $2k/\eps (1+o(1))$ & 
   $O((ndk+dk^3)\eps^{-2/3})$ & 
   \pbox{4cm}{$P_{2k/\eps(1+o(1))}$, $\eps$\textsf{R}}\\ 
    \cline{1-4}
    \end{tabular}
\end{table}

\end{document}